\documentclass[preprint,1p]{elsarticle}
\usepackage{amsmath}

\usepackage{graphicx}
\usepackage{amssymb, amsmath}
\usepackage{bm}
\usepackage[amsmath,thmmarks]{ntheorem}
\usepackage{color,CJK}
\usepackage{CJK}
\usepackage[ruled,lined,algonl]{algorithm2e}
\usepackage{dcolumn}
\usepackage{amsmath}

\theoremseparator{.}

\newtheorem{theorem}{Theorem}[section]
\newtheorem{proposition}[theorem]{Proposition}
\newtheorem{lemma}[theorem]{Lemma}

\theorembodyfont{\rmfamily}
\newtheorem{definition}{Definition}[section]

\theoremstyle{nonumberplain}
\newtheorem{problema}{Problem A}
\newtheorem{problemb}{Problem B}

\theoremsymbol{$\square$}
\newtheorem{proof}{Proof}

\newcommand{\uvar}{\bm{u}}
\newcommand{\xvar}{\bm{x}}

\DeclareMathOperator{\ini}{ini}

\DeclareMathOperator{\rem}{rem}

\DeclareMathOperator{\res}{res}
\DeclareMathOperator{\bp}{BP}
\DeclareMathOperator{\discr}{discr}
\DeclareMathOperator{\nz}{NZero}

\newcommand{\alert}[1]{\emph{#1}}

\newcommand{\pset}[1]{\mathcal{#1}}
\newcommand{\sas}[1]{\mathbb{#1}}

\newcommand{\bases}[1]{\langle #1 \rangle}

\newcommand{\val}[1]{\overline{#1}}

\newcommand{\rnum}{\mathbb{R}}
\newcommand{\cnum}{\mathbb{C}}

\newcommand{\qnum}{\mathbb{Q}}

\newcommand{\grobner}{Gr\"{o}bner }

\newcommand{\step}[1]{\medskip\noindent\textbf{STEP #1}}
\def\ste#1{\medskip\noindent{\sc Step #1}}

\newcommand{\rli}[1]{}

\DeclareMathOperator{\zero}{Zero}
\newcommand{\pss}{\pset{P}}
\newcommand{\tss}{\pset{T}}
\newcommand{\sss}{\pset{S}}

\journal{arXiv}

\begin{document}

\begin{frontmatter}



\title{Computing Equilibria of Semi-algebraic Economies Using Triangular Decomposition and Real Solution Classification}


\author[Beihang]{Xiaoliang Li\corref{corr}}
\ead{xiaoliangbuaa@gmail.com} \cortext[corr]{Corresponding author.
Mobile: +86 13751465512. Fax: N/A. Postal address: School of Computer, Dongguan University of Technology,
Songshan Lake, Dongguan, Guangdong 523808, China}

\author[France]{Dongming Wang}
\ead{dongming.wang@lip6.fr}

\address[dgut]{School of Computer, Dongguan University of Technology, Dongguan, Guangdong 523808, China}
\address[France]{Laboratoire d'Informatique de Paris 6, CNRS
-- Universit\'e Pierre et Marie Curie, 4 place Jussieu -- BP 169,
75252 Paris cedex 05, France}

\begin{abstract}
In this paper, we are concerned with the problem of determining the
existence of multiple equilibria in economic models. We propose a
general and complete approach for identifying multiplicities of
equilibria in semi-algebraic economies, which may be expressed as
semi-algebraic systems. The approach is based on triangular
decomposition and real solution classification, two powerful tools
of algebraic computation. Its effectiveness is illustrated by two
examples of application.
\end{abstract}


\end{frontmatter}

\section{Introduction}

The equilibria of an economy are states where the
quantity demanded and the quantity supplied are balanced. In other
words, the values of the variables at the equilibria in the economic
model remain stable (when there is no external influence). For
example, a market equilibrium refers to a condition under which a
market price is established through competition such that the amount
of goods or services sought by buyers is equal to that produced by
sellers. Equilibrium models have been used in various branches of
economics such as macroeconomics, public finance, and international
trade \cite{m90t}.

When analyzing equilibrium models, economists usually assume the
global uniqueness of competitive equilibria. However, the
rationality of this assumption is not yet convincing. For instance,
in ``realistically calibrated'' models, it is still an open question
whether or not the phenomenon of multiple equilibria is likely to
appear. For this question, Gjerstad \cite{g96m} has achieved some
results: he pointed out that the multiplicity of equilibria is
prevalent in a pure exchange economy which has CES utility functions
with elasticities of substitution above 2. Moveover, from a
practical point of view, sufficient assumptions for the global
uniqueness of competitive equilibria are usually too restrictive to
be applied to realistic economic models.

In the economic context, the multiplicity of equilibria of an
economy refers to the number of equilibria of the economic model.
Detecting the multiplicities of equilibria  of economies is an
important issue, as multiplicities may
cause serious mistakes in the analysis of economics models and the
prediction of economic trends. Moreover, that the known sufficient
conditions for uniqueness are not satisfied does not imply that
there must be several competitive equilibria. This means that the
existing theories and results for the models may remain useful when
the sufficient conditions are not satisfied.

Traditional approaches for computing equilibria are almost all based
on numerical computation. They have several shortcomings: first,
numerical computation may encounter the problem of instability,
which could make the results completely useless; second, most
numerical algorithms only search for a single equilibrium and are
nearly infeasible for multiplicity detection. Thus it is desirable
to develop methods which can detect exactly all the equilibria of
applied economic models.

Recently, Kubler and Schmedders \cite{k10c} have considered a
special kind of standard finite Arrow--Debreu exchange economies
with semi-algebraic preferences, which are called
\emph{semi-algebraic exchange economies}. Following the terminology
used by Kubler and Schmedders, by \alert{semi-algebraic economies}
we mean economic models (including for example competitive models
and equilibrium models with strategic interactions) whose equilibria
can be described as real solutions of semi-algebraic systems, say of
the form
\begin{equation}\label{semi-alg}
\left\{
\begin{array}{l}
P_1(u_1, \ldots, u_d,x_1, \ldots, x_n)=0,\\
\qquad\qquad\qquad\vdots \\
P_n(u_1, \ldots, u_d, x_1, \ldots, x_n)=0,\vspace{4pt}\\
Q_1(u_1, \ldots, u_d,x_1, \ldots, x_n)\lessgtr 0,\\
\qquad\qquad\qquad\vdots \\
Q_r(u_1, \ldots, u_d, x_1, \ldots, x_n)\lessgtr 0,
\end{array}
\right.
\end{equation}
where the symbol $\lessgtr$ stands for any of $>$, $\geq$, $<$,
$\leq$, and $\neq$, and $P_i, Q_j$ are polynomials over the field
$\rnum$ of real numbers, with $u_1, \ldots, u_d$ as their parameters
and $x_1, \ldots, x_n$ as their variables. Note that systems from
realistic economies should be zero-dimensional, i.e., their zeros
$(\bar{x}_1, \ldots, \bar{x}_n)$ must be finite in number under any
meaningful specialization of the parameters $u_1, \ldots, u_d$.

Thus for semi-algebraic economies the problem of computing
equilibria may be reduced to that of dealing with the semi-algebraic
system \eqref{semi-alg}. For example, the multiplicity of equilibria
can be detected by determining whether or not the corresponding
system \eqref{semi-alg} has multiple real solutions. This problem
has been solved partially by Kubler and Schmedders \cite{k10c,k10t}
using \grobner bases. The main idea that underlies the remarkable
work of Kubler and Schmedders is to use the method of \grobner bases
to transform the equation part of system \eqref{semi-alg} into an
equivalent set of new equations in a much simpler form, where only
one equation, say $G_1=0$, is nonlinear, yet it is univariate, and
to count the real solutions of the equation part by using Strum's
sequence of $G_1$.

On the other hand, Datta \cite{d10f,d03u} has compared the methods
of \grobner bases and homotopy continuation for computing all
totally mixed Nash equilibria in games. Chatterji and Gandhi
\cite{c10a} have applied computational Galois theory to the problem
of computing Nash equilibria of a subclass of generic finite normal
form games, i.e., the rational payoff games with irrational
equilibria.

The work presented in this paper is based on our observation that
triangular decomposition of polynomial systems \cite{w01e,h03n} and
real solution classification of semi-algebraic systems
\cite{y01c,y05r} may serve
as a good alternative to \grobner bases and Strum sequences for the
computation of equilibria of semi-algebraic economies. This
alternative may lead to new approaches that are theoretically more
general and practically more effective than the approaches developed
by Kubler, Schmedders, and others. The aim of the present paper is
to propose one such approach, which is general and complete, for
identifying the multiplicity of equilibria in semi-algebraic
economies. The proposed approach takes inequalities into
consideration and can give a tighter bound or
even precise number of equilibria, depending on whether the economy
is exactly described by \eqref{semi-alg}, than the existing
approaches mentioned above, which only compute an upper bound for
the number of equilibria because inequality constraints from
realistic economies are usually ignored for simplicity.

The key step of our approach is to decompose the semi-algebraic
system in question into several triangularized semi-algebraic
systems, with the total number of solutions unchanged. Consider for
example the system
\begin{equation}\label{ex:tri-exp}
\left\{
\begin{array}{l}\smallskip
P_1=x_2x_3-1=0,\\ \smallskip
P_2=x_4^2+x_1x_2x_3=0,\\ \smallskip
P_3=x_1x_2x_4+x_3^2-x_2=0,\\
P_4=x_1x_3x_4-x_3+x_2^2=0.
\end{array}
\right.
\end{equation}
Under the variable ordering $x_1<\cdots<x_4$, triangular
decomposition of the polynomial set $\pss=\{P_1, \ldots, P_4\}$
results in two \emph{triangular sets}
\begin{equation*}
  \begin{split}
    \tss_1=[x_1^3+4,x_2^3+1,x_2x_3-1,2\,x_4+x_1^2],\quad
     \tss_2=[x_1,x_2^3-1,x_2x_3-1,x_4],
  \end{split}
\end{equation*}
such that the union of the zero sets of $\tss_1$ and $\tss_2$ is
same as the zero set of $\pss$. The zeros of the triangular sets
$\tss_1$ and $\tss_2$ may be computed successively, which is easier
than computing the zeros directly from $\pss$. Triangular
decomposition as such is used in the first stage of our approach to
preproccess the equation part of \eqref{semi-alg}.

The rest of the paper is structured as follows. In Section
\ref{sec:no-par}, we first show how to count equilibria of
semi-algebraic economies without parameters by means of a simple
example and then describe a complete method for the counting. In
Section \ref{sec:par}, a method based on real solution
classification is presented to deal with semi-algebraic economies
with parameters. In Section \ref{sec:exp}, we demonstrate the
effectiveness of our methods using two examples of application. The
paper is concluded with some remarks in Section \ref{sec:con}.

\section{Economies Without Parameters}\label{sec:no-par}


From now on we denote by $\uvar$ and $\xvar$ the parameters $u_1,
\ldots, u_d$ and the variables $x_1, \ldots, x_n$ respectively in
system \eqref{semi-alg}. In this section, we consider the simpler
case when $\uvar$ do not occur in \eqref{semi-alg}.

\begin{problema}
Assume that the parameters $\uvar$ are not present in system
\eqref{semi-alg}. Count all the distinct real solutions of
\eqref{semi-alg}.
\end{problema}

The method that we will present for solving this problem extends the
approach of Kubler and Schmedders \cite{k10c,k10t}. It can
systematically handle economies with inequality conditions (which
are fairly prevalent in practical applications).

\subsection{Triangular Decomposition Revisited}

We recall some standard notations and algorithms used for triangular
decomposition of polynomial systems, which play a fundamental role
in our approach to be proposed.

Let the variables be ordered as $x_1 < \cdots < x_n$. An ordered set
$[T_1, \ldots, T_r]$ of non-constant polynomials is called a
\emph{triangular set} if the leading variable of $T_i$ is smaller
than that of $T_j$ for all $i<j$, where the \emph{leading variable}
of $T_i$ is the variable with biggest index occurring in $T_i$. For
example, $[x_1-2,(x_1^2-4)x_3^3-x_2]$ is a triangular set.

Let $\pset{P}$ and $\pset{Q}$ be two sets of multivariate
polynomials with coefficients in the field $\qnum$ of rational
numbers. We denote by $\zero(\pset{P})$ the set of all common zeros
(in some extension field of $\qnum$) of the polynomials in $\pss$
and by $\zero(\pset{P}/\pset{Q})$ the subset of $\zero(\pset{P})$
whose elements do not annihilate any polynomial in $\pset{Q}$.

Any multivariate polynomial can be viewed as a univariate polynomial
in its leading variable. We use $\ini(\pss)$ to denote the set of
leading coefficients of all the polynomials in $\pss$, viewed as
univariate polynomials in their leading variables. Such leading
coefficients are called \emph{initials}.


\begin{theorem}\label{thm:td}
There are algorithms which can decompose any given polynomial set
$\pss$ into finitely many triangular sets
$\pset{T}_1,\ldots,\pset{T}_k$ with different properties such that
\begin{equation}\label{eq:tri-decom}
\zero(\pss)=\bigcup_{i=1}^k \zero(\pset{T}_i/\ini(\pset{T}_i)).
\end{equation}
\end{theorem}

Among the algorithms pointed out by the above theorem, the best
known is Wu--Ritt's algorithm based on the computation of
characteristic sets, developed by Wu \cite{w86b,w86z} from the work
of Ritt \cite{r50d} in differential algebra. For example, the
polynomial set
$$\pset{P}=[xy^2+z^2,xz+y]$$
may be decomposed by using Wu--Ritt's algorithm with $x<y<z$ into
four triangular sets
\begin{equation}\label{ex:charser}
  \begin{split}
    & \pset{T}_1=[(x^3+1)y^2, xz+y],\quad
    \pset{T}_2=[x^2-x+1, xz+y], \\
    & \pset{T}_3=[x+1, z-y], \qquad\qquad\!
      \pset{T}_4=[x, y, z^2].
  \end{split}
\end{equation}

It is not guaranteed that $\zero(\tss/\ini(\tss))\neq\emptyset$ for
all triangular set $\tss$. For example, with
$$\tss=[x^2-u, y^2+2\,xy+u,(x+y)z+1]$$
and $u<x<y<z$, it can be easily proved that
$\zero(\tss/\ini(\tss))=\emptyset$. We can impose additional
conditions to obtain triangular sets of other kinds with better
properties. Typical examples of such triangular sets are
\emph{regular sets} \cite{w00c} (also known as regular chains
\cite{k93g} and proper ascending chains \cite{y94s}), \emph{simple
sets} \cite{w98d}, and \emph{irreducible triangular sets}
\cite{w86b,w01e}.

A triangular set $[T_1, \ldots, T_r]$ is said to be \emph{regular}
or called a \emph{regular set} if no {regular zero of
$\tss_i$} annihilates
the initial of $T_{i+1}$ for all $i=1,\ldots,r-1$, {where
$\tss_i=[T_1,\ldots,T_i]$ and a regular zero of $\tss_i$ is such a
zero of $\tss_i$ in which the variables other than the leading
variables of $T_1,\ldots,T_i$ are not specialized to concrete
values.} For example, {with $u<x<y$ the triangular set
${\hat{\tss}}=[x^2-u^2,(x+u)y+1]$ is not regular, while
$\bar{\tss}=[x^2-u^2,xy+1]$ is. Observe that the initials of the
second polynomials in both $\hat{\tss}$ and $\bar{\tss}$ vanish
at the zeros of $x^2-u^2$ when $u$ is specialized to $0$.} The
reader may refer to \cite{w00c,w01e,h03n} for formal definitions of
regular sets and \cite{k93g,y94s,w00c,m00t} for effective algorithms
that decompose any given polynomial set into finitely many regular
sets.

{A regular set $[T_1,\ldots,T_r]$ is called a \emph{simple
set} or an \emph{irreducible triangular set}, respectively, if every
$T_{i}$ is squarefree or irreducible at any regular zero of
$[T_1,\ldots,T_{i-1}]$ for {$i=1,\ldots,r$}. For example, the regular
set $\bar{\tss}$ given above} is also a simple set, but it is not
irreducible. Simple sets and irreducible triangular sets have many
nice properties (of which some are about their saturated ideals, see
\cite{w01e,h03n,l10d}). Algorithms for decomposing polynomial sets
into simple sets or irreducible triangular sets may be found in
\cite{w86b,w98d,w01e,l10d,m13d}.

The algorithms for triangular decomposition proposed by the second
author \cite{w93e,w00c,w98d} appear to be more general than other
available ones. They can be used to decompose any given polynomial
system $[\pset{P},\pset{Q}]$ into finitely many triangular systems
$[\pset{T}_1,\pset{S}_1],\ldots,[\pset{T}_k,\pset{S}_k]$ with
different properties such that
\begin{equation}\label{eq:tri-sys}
\zero(\pss/\pset{Q})=\bigcup_{i=1}^k \zero(\pset{T}_i/\pset{S}_i),
\end{equation}
where $[\pset{T}_i,\pset{S}_i]$ could be fine triangular systems
\cite{w93e}, regular systems \cite{w00c}, or simple systems
\cite{w98d}, corresponding to triangular sets, regular sets, or
simple sets respectively. The interested reader may consult the
above-cited references for formal definitions, properties, and
algorithms.

Triangular decomposition discussed above may be effectively used in
our approach for counting real solutions of semi-algebraic systems.
Note that for the decomposition \eqref{eq:tri-decom} or
\eqref{eq:tri-sys}, there is no guarantee that
$$\zero(\pset{T}_i/\ini(\pset{T}_i))\cap \zero(\pset{T}_j/\ini(\pset{T}_j))=\emptyset
~~\text{or}~~\zero(\pset{T}_i/\sss_i)\cap \zero(\pset{T}_j/\sss_j)=
\emptyset$$ for $i\neq j$. For the triangular sets in
\eqref{ex:charser}, it is easy to verify that
$(\frac{1+\rm{i}\sqrt{3}}{2},0,0)$ is in both
$\zero(\pset{T}_1/\ini(\pset{T}_1))$ and
$\zero(\pset{T}_2/\ini(\pset{T}_2))$. This problem may cause some
trouble for counting distinct real zeros, but it can be solved,
e.g., by using the technique given in \cite{y92c} {(see
also \cite{w98d}).}

A triangular set in which all polynomials other than the first are
linear with respect to (w.r.t.) their corresponding leading
variables is said to be \alert{quasi-linear}. For example, in
\eqref{ex:charser} $\tss_1$, $\tss_2$, and $\tss_3$ are all
quasi-linear, but $\tss_4$ is not. The (real) zeros of quasi-linear
triangular sets may be determined by analyzing essentially the first
polynomials in the triangular sets. A triangular system
$[\tss,\sss]$ is said to be \emph{quasi-linear} if $\tss$ is
quasi-linear. Quasi-linearization is a key step in transforming an
arbitrary semi-algebraic system into an equivalent semi-algebraic
system in which the equation polynomials form a quasi-linear
triangular set.

\begin{theorem}\label{th:qualin}
  Let $\tss=[T_1(\uvar,y_1),\dots,T_r(\uvar,y_1,\dots,y_r)]$
  be a regular set in $\qnum[\uvar,y_1,\dots,y_r]$ and
  $c_2,\dots,c_r$ be a sequence of $r-1$ randomly chosen {integers}.
  Then the polynomial set $\tss^*$ obtained from $\tss$ by replacing
  $y_1$ with $y_1+c_2y_2+\cdots+c_ry_r$ can be decomposed over
  $\qnum(\uvar)$, with probability $1$, into finitely many quasi-linear triangular sets $\tss_1,\ldots,\tss_k$
  w.r.t.\ the variable ordering $y_1<\cdots<y_r$, such that
\begin{equation*}
\zero(\tss^*)=\bigcup_{i=1}^k \zero(\pset{T}_i/\ini(\pset{T}_i)).
\end{equation*}
\end{theorem}
\begin{proof}
  Let $\tss^*$ be decomposed into $k$ simple sets
  $$\pset{T}_i^*=[T_{i1}^*(\uvar,y_1),\dots,T_{ir}^*(\uvar,y_1,\dots,y_r)],\quad
  i=1,\ldots, k,$$
according to Theorem~\ref{thm:td}. By means of normalization (using,
e.g., \cite[Algorithm 4]{l08a}) and pseudo-division, each
$\pset{T}_i^*$ may be transformed into a normal and reduced simple
set
$\pset{T}_i'=[T_{i1}'(\uvar,y_1),\dots,T_{ir}'(\uvar,y_1,\dots,y_r)]$
  such that $\bases{\pset{T}_i^*}=\bases{\pset{T}_i'}$ and the degree of $T_{ij}^*$ in
  $y_j$ remains unchanged (see \cite{w01e} for the definitions of \emph{normal} triangular set and \emph{reduced} triangular set),
where $\bases{\pset{T}_i^*}$ denotes the ideal in
$\qnum(\uvar)[y_1,\dots,y_r]$ generated by $\pset{T}_i^*$.  Let
  $$\pset{T}_i=[T_{i1}'/\ini(T_{i1}'),\ldots,T_{ir}'/\ini(T_{ir}')].$$
  Note that under the lexicographical term order, the leading monomials of any two different polynomials in
 $\pset{T}_i$ are relatively prime.
 Thus $\pset{T}_i$ is a \grobner basis of
 $\bases{\pset{T}_i'}$ by Proposition~4 in \cite[section 2.9]{c97i}.
  The monicness of the
 polynomials in $\pset{T}_i$ is obvious. As $\pset{T}_i$ is reduced,
 $\pset{T}_i$ is the reduced \grobner basis of $\bases{\pset{T}_i'}$.
 Furthermore, from \cite[Theorem 3.3]{l10d} we know that $\bases{\pset{T}_i}$ is a radical ideal
 (because $\tss_i$ is a simple set).

{Let
$(\val{y}_{11},\ldots,\val{y}_{r1}),\ldots,(\val{y}_{1s},\ldots,\val{y}_{rs})$
be all the (distinct) zeros of $\bases{\pset{T}}$ in the algebraic
closure $K$ of $\qnum(\uvar)$. Then for any $\mu\neq \nu$,
$$H(z_r,\ldots,z_2)=(\val{y}_{r\mu}-\val{y}_{r\nu})z_r+\cdots+(\val{y}_{2\mu}-\val{y}_{2\nu})z_2+(\val{y}_{1\mu}-\val{y}_{1\nu})=0$$
defines a  hyperplane or an empty set of $K^{r-1}=\{(z_r,\ldots,z_2)|\, z_i\in K\}$.
As $c_r,\dots,c_2$ are randomly chosen integers, the probability that
$H(c_r,\ldots,c_2)\neq 0$ (i.e., {$(c_r,\ldots,c_2)$ is not among the integer points in the hyperplane}) is $1$. Note that
$$\val{y}_{1\mu}+c_2\val{y}_{2\mu}+\cdots+c_r\val{y}_{r\mu}=y_{1\mu}^*,\quad
\val{y}_{1\nu}+c_2\val{y}_{2\nu}+\cdots+c_r\val{y}_{r\nu}=y_{1\nu}^*$$
are the $y_1$-coordinates of two zeros of some $\bases{\pset{T}_i}$
and $\bases{\pset{T}_j}$. It is thus with probability $1$ that
$y_{1\mu}^*-y_{1\nu}^*\neq 0$ for any $\mu\neq \nu$. Therefore, with
probability $1$ the $y_1$-coordinates of the zeros of all
$\bases{\pset{T}_i}$ are distinct.}

By the Shape Lemma \cite{s02s}, each $\pset{T}_i$ must be of the form
 $[G_{i1}(y_1),y_2-G_{i2}(y_1),\ldots,y_r-G_{ir}(y_1)]$,
 where $G_{ij}$ are polynomials over $\qnum(\uvar)$. Thus $\pset{T}_i$
 is quasi-linear and the proof is complete.
\end{proof}

Triangular sets produced by triangular decomposition are often, but
not always, quasi-linear. Among the four triangular sets in
\eqref{ex:charser}, only $\tss_4$ is not quasi-linear. {One
may obtain quasi-linear triangular sets by means of linear
transformation with randomly chosen integers $c_2,\dots,c_r$
according to the above theorem. As the probability of success with
one trial is $1$, the} quasi-linearization technique is effective.

\subsection{Illustrative Example}\label{sec:ex-no-par}

Before describing the method, we provide a simple example to
illustrate our general approach for solving Problem A. Consider the
semi-algebraic system
\begin{equation*}\label{ex:all-part}
\left\{
\begin{array}{l}
x^3-20\,y^2=0,\\
y^2-2\,x-1=0, \\
x-y\neq 0,\\
2\,x-y\geq 0,\\
y >0.
\end{array}
\right.
\end{equation*}
It is easy to see that the number of distinct real solutions of this
system is equal to the sum of those of the following two systems:
\begin{equation}\label{ex:all-pre}
\left\{
\begin{array}{l}
x^3-20\,y^2=0,\\
y^2-2\,x-1=0, \\
x-y\neq 0,\\
2\,x-y = 0,\\
y >0,
\end{array}
\right.
\end{equation}
\begin{equation}\label{ex:all-part}
\left\{
\begin{array}{l}
x^3-20\,y^2=0,\\
y^2-2\,x-1=0, \\
x-y\neq 0,\\
2\,x-y> 0,\\
y >0.
\end{array}
\right.
\end{equation}
These two systems may be treated similarly, so we only consider
\eqref{ex:all-part} in what follows. The number of distinct real
solutions of the system can be determined in five steps.

\ste{A1.} Let the variables be ordered as $x<y$ and decompose the
set of equation polynomials
\begin{equation*}\label{ex:pol-part}
\pset{F}=\{x^3-20\,y^2,~ y^2-2\,x-1\}
\end{equation*}
into triangular sets. The process of triangular decomposition for
$\pset{F}$ is trivial and one can easily obtain a triangular set
\begin{equation}\label{eq:sim-sys}
\tss=[x^3-40\,x-20, ~y^2-2\,x-1]
\end{equation}
such that $\zero(\tss)=\zero(\pset{F})$.

Substituting $x$ in $\tss$ by $x+y$ and decomposing the resulting
set $\tss^*$ into simple sets under the same ordering $x<y$, we
obtain $\tss_1=[T_1, T_2]$ with
\begin{equation*}\label{eq:lin-sys}
\begin{split}
 &T_1=x^6-83\,x^4-360\,x^3+1083\,x^2+1320\,x+359,\quad T_2=Iy+J,\\
 &I=3\,x^2+8\,x-35,\quad J=x^3+6\,x^2-33\,x-18.
\end{split}
\end{equation*}
Now $T_1$ is univariate, $T_2$ is linear in $y$, and
$\zero(\tss_1/\{I\})=\zero(\tss^*)$.

The linear transformation $x\rightarrow x+y$ above transforms the
inequality constraints $x-y\neq 0$, $2\,x-y> 0$, and $y >0$ in
\eqref{ex:all-part} into $x\neq 0$, $2\,x+y > 0$, and $y>0$
respectively. Hence the number of real solutions of
\eqref{ex:all-part} is equal to that of the following system
\begin{equation*}
\left\{
\begin{array}{l}
T_1=0,\\
T_2=0, \\
I\neq 0,\\
x\neq 0,\\
2\,x+y > 0,\\
y >0.
\end{array}
\right.
\end{equation*}

\ste{A2.} Solving $T_2=0$ for $y$ yields $y=-J/I$. Substituting this
solution into the inequality constraints in the above system, we
obtain the following system
\begin{equation*}
\left\{
\begin{split}
&T_1=0,\\
&I\neq 0,\\
&x\neq 0,\\
&2\,x-J/I > 0,\\
&-J/I >0.
\end{split}
\right.
\end{equation*}

The constraint $-J/I>0$ can be replaced by the equivalent inequality
$G=-JI>0$. Similarly, $2\,x-J/I > 0$ can be replaced by
$H=2\,xI^2-JI>0$.

\ste{A3.} Further replace the constraints $G>0$ and $H>0$
respectively by
$$G'=\rem(G,F)=-3\,x^5-26\,x^4+86\,x^3+528\,x^2-1011\,x-630>0,$$
and
$$H'=\rem(H,F)=15\,x^5+70\,x^4-206\,x^3-592\,x^2+1439\,x-630>0,$$
where $\rem(P,F)$ denotes the \emph{remainder} of $P$ divided by
$F$. Then the problem is reduced to counting the real solutions of
the following semi-algebraic system in a single variable $x$:
\begin{equation}\label{ex:uni}
\left\{
\begin{array}{l}
F=0,\\
H'>0,\\
G'>0,
\end{array}
\right.
\end{equation}
where $F$ and $G'$ have no common zeros, and so do $F$ and $H'$.

\ste{A4.} In order to count the real solutions of \eqref{ex:uni}, we
isolate the real zeros of $G'\cdot H'$ by rational intervals using
an available algorithm. For instance, application of the modified
Uspensky algorithm \cite{c83r} may yield the following sorted
sequence of intervals
\begin{eqnarray*}
&[-16, -8],~ [-5, -5],~ [-9/2, -4],~ [-1, -1/2],~ [1/2, 3/4],~ [1,
3/2],~ [2, 5/2],~ [3, 4].
\end{eqnarray*}
These closed intervals do not intersect with each other, and each of
them contains one and only one distinct real zeros of $G'$ or $H'$.
Moreover, by Sturm's theorem \cite{s02s} or simply using the $\sf
sturm$ function in Maple, we can prove that all the intervals cover
no real zero of $F$.

\ste{A5.} The real zeros of $F$ must be in
\begin{eqnarray*}
&(-\infty, -16),~ (-8, -5),~ (-5, -9/2),~ (-4, -1),\\[2pt]
&(-1/2, 1/2),~ (3/4, 1),~ (3/2, 2),~ (5/2, 3),~ (4, +\infty),
\end{eqnarray*}
the complement of the intervals given in the preceding step. In each
of these open intervals, the signs of $G'$ and $H'$ are invariant,
and can be determined by simply testing them at a sample point in
the interval. For example, to determine the sign of $G'$ on
$(-\infty, -16)$, we may compute $G'(-16-1)$ to get $1834656>0$.
Thus $G'$ is positive at every point in $(-\infty, -16)$. Proceeding
in this way for each interval, we can conclude that $G'$ and $H'$
are positive on and only on
\begin{equation*}
(-5, -9/2),\quad (5/2, 3).
\end{equation*}

Finally, applying the $\sf sturm$ function to count the real zeros
of $F$ on the above two open intervals, one finds that their numbers
are $0$ and $1$ respectively. In conclusion, the original
semi-algebraic system \eqref{ex:all-part} has only one real
solution.

\subsection{General Method}\label{sec:par-alg}

In this section, we formulate the steps of the illustrative example
to a general method for counting equilibria of semi-algebraic
economies without parameters. As we have explained, the problem may
be reduced to counting distinct real solutions of \eqref{semi-alg},
where the parameters $\uvar$ are not present. Since an inequality
constraint like $P\geq 0$ can be split into $P=0$ or $P>0$, we only
need to consider semi-algebraic systems of the form
\begin{equation}\label{eq:stand-sys}
\left\{
\begin{array}{l}\smallskip
F_1(\xvar)=0,\ldots,F_n(\xvar)=0,\\ \smallskip
N_1(\xvar)\neq 0,\ldots,N_s(\xvar)\neq 0,\\
P_1(\xvar)> 0,\ldots,P_t(\xvar)> 0.
\end{array}
\right.
\end{equation}
Let $\pset{F}=\{F_1,\ldots,F_n\}$ and $\pset{N}=\{N_1,\ldots,N_s\}$
and the variables be ordered as $x_1<\cdots<x_n$. The method that
provides an effective and complete solution to the problem consists
of the following steps.

  \step{A1.} Decompose the polynomial system $[\pset{F}, \pset{N}]$
  into finitely many triangular systems $[\tss_1,\sss_1],\ldots,[\tss_k,\sss_k]$ such that
  $$\zero(\pset{F}/\pset{N})=\bigcup_{i=1}^k\zero(\tss_i/\sss_i)~~\text{and}~~\zero(\tss_i/\sss_i)\cap\zero(\tss_j/\sss_j)=\emptyset,~i\neq j.$$
We may assume that all $\tss_i$ are quasi-linear, for otherwise
$\tss_i$ may be made quasi-linear by the quasi-linearization process
using appropriate linear transformations (see
Theorem~\ref{th:qualin} and remarks thereafter). Then the problem is
reduced to counting the distinct real zeros in each
$\zero(\tss_i/\sss_i)$ which satisfy $P_l>0$.

  \step{A2.} For each $i$, let
  $$\tss_i=[T_{i1}(x_1),\dots,T_{in}(x_1,\ldots,x_n)]$$
  and $S_i$ be the product of all polynomials in $\sss_i$.
  Solve $T_{ij}=0$ for $x_j$, $j=n,\ldots,2$, and
  substitute the solutions successively into $S_i$ and $P_l$ to obtain
  rational functions $A_{i}/A'_{i}$ and $B_{il}/B'_{il}$ respectively,
  where $A_{i},A'_{i},B_{il},B'_{il}$ are all univariate polynomials in $x_1$.
  The problem is further reduced to counting the distinct real solutions of the semi-algebraic system
  \begin{equation}\label{sysuni}
\left\{
\begin{array}{l}\smallskip
T_{i1}=0,\\ \smallskip
A_{i}\neq 0,\\
B^*_{il}=B_{il}\cdot B'_{il}>0,~l=1,\ldots,t,
\end{array}
\right.
\end{equation}
in one variable $x_1$ for $i=1,\ldots, k$.

  \step{A3.} Simplify system \eqref{sysuni}, for example, by removing from $T_{i1}$ its common
  factors with every $A_{i}$ and $B^*_{il}$ to obtain $T'_{i1}$ and
  then replacing each $B^*_{il}$ with $C_{il}=\rem(B^*_{il},T'_{i1})$.
  In this way, we arrive at the system
  \begin{equation}\label{systuni}
  T'_{i1}=0,\quad C_{il}>0,~l=1,\ldots,t,
  \end{equation}
which is equivalent to \eqref{sysuni} for $i=1,\ldots, k$.

  \step{A4.} Isolate the real zeros of each $C_{il}$ in \eqref{systuni} using, e.g., the
  modified Uspensky algorithm \cite{c83r} to obtain a sequence of closed intervals
  $[a_1,b_1],\ldots,[a_m,b_m]$, such that
  \begin{itemize}
    \item $a_i,b_i$ are all rational numbers,
    \item $a_1\leq b_1<a_2\leq b_2<\cdots<a_m\leq b_m$,
    \item $[a_i,b_i]\cap[a_j,b_j]=\emptyset~\text{for}~i\neq j$,
    \item each $[a_i,b_i]$ contains one and only one real zero of some $C_{il}$,
    \item every $[a_i,b_i]$ covers no real zero of $T'_{i1}$.
  \end{itemize}

 \step{A5.}
 In each connected subset of the complement $(-\infty,a_1)\cup (b_1,a_2)\cup\cdots\cup
 (b_m,+\infty)$ of the above intervals, the sign of each $C_{il}$ is invariant and
 can be determined by computing the value of $C_{il}$ at a sample point in the subset.
 From the complement, select open intervals on which all the $C_{il}$ are positive.
 Finally, apply Sturm's theorem to count the real zeros of $T'_{i1}$ on those selected
 intervals, and sum them up.

\medskip
The correctness of the above method is quite obvious. Although the
method is effective and may be easily understood and implemented,
the process of quasi-linearization is time-consuming, in particular
for systems with polynomials of high degree. Xia and Hou \cite{x02c}
proposed a direct method, which has the same functionality as ours,
but does not need to make triangular sets quasi-linear. The method
of Xia and Hou is also based on triangular decomposition and it
works by recursively computing the so-called near roots of
polynomials in triangular sets.

\section{Economies with Parameters}\label{sec:par}

If parameters appear in system \eqref{semi-alg}, the number of real
solutions of the system may change along with the variation of
parameters. In this case, the method presented in the previous
section cannot be directly applied to the analysis of equilibria.
The problem of our concern is formulated as follows.

\begin{problemb}
Assume that the parameters $\uvar$ are present in system
\eqref{semi-alg}. For any given non-negative integer $k$, determine
the condition on $\uvar$ for system \eqref{semi-alg} to have exactly
$k$ distinct real solutions.
\end{problemb}

This is the problem of \emph{real solution classification} for
\eqref{semi-alg}. We show how to solve the problem by first using
triangular decomposition with quasi-linearization to reduce
\eqref{semi-alg} to semi-algebraic systems in a single variable with
parameters and then determining the numbers of distinct real
solutions of such systems at sample points in regions of the
parameter space decomposed by the border polynomials of the
semi-algebraic systems in one variable.


\subsection{Preliminaries}

The main purpose of this section is to define the border polynomial
of a semi-algebraic system in one variable. We first introduce some
notations.

Let $$F=\sum_{i=0}^ma_i\,x^i,\quad G=\sum_{j=0}^lb_j\,x^j$$ be two
univariate polynomials in $x$ with coefficients $a_i,b_j$ in the
field $\cnum$ of complex numbers, and $a_m,b_l\neq 0$. The
determinant
\begin{equation*}\label{eq:sylmat}
 \begin{array}{c@{\hspace{-5pt}}l}
 \left|\begin{array}{cccccc}
a_m & a_{m-1}& \cdots   & a_0   &        &       \\
           & \ddots   & \ddots&    \ddots    &\ddots&   \\
         &          & a_m   & a_{m-1}&\cdots& a_0 \\ [5pt]
b_l & b_{l-1}& \cdots   &  b_0 &    &         \\
            & \ddots   &\ddots &   \ddots     &\ddots&       \\
         &   &    b_{l}     & b_{l-1} & \cdots &  b_0
\end{array}\right|
& \begin{array}{l}\left.\rule{0mm}{8mm}\right\}l\\
\\\left.\rule{0mm}{8mm}\right\}m
\end{array}
\end{array}
\end{equation*}
is called the \emph{Sylvester resultant} (or simply
\emph{resultant}) of $F$ and $G$, and denoted by $\res(F,G)$. The
following lemma reveals the relation between the common zeros and
the resultant of two polynomials.

\begin{lemma}[\cite{m93a}]\label{lem:res-com}
 Two univariate polynomials $F$ and $G$ have common zeros in $\cnum$ if and only if $\res(F,G)=0$.
\end{lemma}

Let ${\rm d}F/{\rm d}x$ denote the derivative of $F$ w.r.t.\ $x$.
The resultant of $F$ and ${\rm d}F/{\rm d}x$,  $\res(F,{\rm d}F/{\rm
d}x)$, is called the \emph{discriminant} of $F$ and denoted by
$\discr(F)$. The following proposition may be easily proved by
definition.

\begin{proposition}[\cite{m93a}]
A univariate polynomial $F$ has multiple zeros in $\cnum$ if and
only if $\discr(F)=0$.
\end{proposition}

Denote by $\nz(*)$ the number of distinct real zeros or solutions of
$*$, where $*$ may be a polynomial, a polynomial set, or a
semi-algebraic system. Consider the following semi-algebraic system
in $x$ with parameters $\uvar$:
\begin{equation*}\label{ex:uni-pq}
\sas{S}=\left\{
\begin{array}{l}\smallskip
P(\uvar,x)=0,\\
Q_1(\uvar,x)>0,\ldots,Q_s(\uvar,x)>0,
\end{array}
\right.
\end{equation*}
where $P(\uvar,x)=\sum_{i=0}^m a_i(\uvar)\,x^i$. It is easy to see
that $\nz(P)$ may change when the leading coefficient $a_m(\uvar)$
or the discriminant $\discr(P)$ goes from non-zero to zero and vice
versa. Moreover, if $\res(P,Q_i)$ goes across zero, then the zeros
of $P$ will pass through the boundaries of the intervals determined
by $Q_i>0$, which means that $\nz(\sas{S})$ may change. This
motivates the following definition.

\begin{definition}[Border Polynomial]
The product
$$a_m(\uvar)\cdot\discr(P)\cdot\prod_{i=1}^s\res(P,Q_i)$$ is called
the \emph{border polynomial} of $\sas{S}$ and denoted by
$\bp(\sas{S})$.
\end{definition}

Based on the above discussions, the proof of the following theorem
is obvious.

\begin{theorem}\label{thm:main}
  Let $A, B$ be two points in the space of parameters $\uvar$, which do not annihilate $\bp(\sas{S})$.
  If there exists a real path $C$ from
  $A$ to $B$ such that $C\cap\zero(\bp(\sas{S}))=\emptyset$, then $\nz(\sas{S}|_{A})=\nz(\sas{S}|_{B})$.
\end{theorem}

\subsection{Illustrative Example}

Now we use an example to explain our general approach for solving
Problem B. Consider the semi-algebraic system
\begin{equation}\label{ex:with-par}
\left\{
\begin{array}{l}\smallskip
x^3-uy^2=0,\\ \smallskip y^2-2\,x-1=0, \\ \smallskip
x-y\neq 0,\\
y+s >0,
\end{array}
\right.
\end{equation}
where $s,u\in \rnum$ are parameters. The following four steps permit
us to decompose the parameter space into regions such that on each
of them the number of distinct real solutions of \eqref{ex:with-par}
is invariant and computable.

\ste{B1.} Decomposing $\pss=[x^3-uy^2,y^2-2\,x-1]$ under the
variable ordering $u<x<y$, we obtain three regular systems $[\tss_1,
\{S_1\}], [\tss_2, \emptyset], [\tss_3, \emptyset]$ with
\begin{equation*}
\begin{split}
  &\tss_1=[-x^3+2\,ux+u, -y^2+2\,x+1],\quad S_1=u(32\,u-27),\\
  &\tss_2=[u, x, y^2-1], \\
  &\tss_3=[32\,u-27, 8\,x^2-6\,x-9, -y^2+2\,x+1].
\end{split}
\end{equation*}
It follows that for any given values $\val{s},\val{u}$ of the
parameters $s,u$:
\begin{itemize}
  \item if $S_1|_{(\val{s},\val{u})}\neq 0$, then
  $\zero(\pss|_{(\val{s},\val{u})})=\zero(\tss_1|_{(\val{s},\val{u})})$;
  \item if $\val{u}=0$, then $\zero(\pss|_{(\val{s},\val{u})})=\zero([x, y^2-1])$;
  \item if $32\,\val{u}-27=0$, then $\zero(\pss|_{(\val{s},\val{u})})=\zero([8\,x^2-6\,x-9, -y^2+2\,x+1])$,
\end{itemize}
where $S_1|_{(\val{s},\val{u})}$, $\pss|_{(\val{s},\val{u})}$, and
$\tss_1|_{(\val{s},\val{u})}$ denote the results of $S_1$, $\pss$,
and $\tss_1$ after substitution of $(s,u)$ by $(\val{s},\val{u})$
respectively.

We consider only the main partition $\{(s,u)|~S_1\neq 0\}$ of the
parameter space, which is of the same dimension as $\rnum^2$.
Triangular systems corresponding to main partitions are called
\emph{main branches} of the triangular decomposition.

Under the same variable ordering $u<x<y$, $\tss_1|_{x=x+y}$ may be
decomposed into four {simple} systems. The only main branch
is $[[T_1, T_2], \{S_2\}]$ with
\begin{equation*}
  \begin{split}
    &T_1=x^6-(4\,u+3)x^4-18\,ux^3+(4\,u^2-26\,u+3)x^2\\
        &~~~~~~~+(4\,u^2-14\,u)x+u^2-2\,u-1,\\
    &T_2=Iy+J,\\
    &I=-3\,x^2-8\,x+2\,u-5,\\
    &J=-x^3-6\,x^2+(2\,u-7)x+u-2,\\
    &S_2=u(32\,u^2-67\,u+64).
  \end{split}
\end{equation*}

\ste{B2.} Solving $T_2 = 0$ for $y$ yields $y = -J/I$. Substituting
this solution into the inequality constraint $y+s>0$ in
\eqref{ex:with-par}, we obtain $-J/I+s>0$, which is equivalent to
$P=(-J+Is)I>0$. On the other hand, the linear transformation
$x\rightarrow x+y$ above transforms the constraint $x-y\neq 0$  into
$x\neq 0$. Thus, when $S_1\neq 0$ and $S_2\neq 0$, the number of
real solutions of system \eqref{ex:with-par} is the same as that of
\begin{equation*}
\left\{
\begin{array}{l}\smallskip
T_1=0,\\ \smallskip
x\neq 0, \\
P>0.
\end{array}
\right.
\end{equation*}

By Lemma \ref{lem:res-com}, $T_1$ and $x$ have no common zero if
$\res(T_1,x)=u^2-2\,u-1\neq 0$. Hence, in the case when
$S_1S_2(u^2-2\,u-1)\neq 0$, the problem is reduced to that of
 real solution classification for the following semi-algebraic system in one variable $x$:
\begin{equation*}\label{ex:uni-par}
\sas{U}=\left\{
\begin{array}{l}\smallskip
T_1=0,\\
P >0.
\end{array}
\right.
\end{equation*}

\ste{B3.}
The border polynomial of $\sas{U}$ is
$$\bp(\sas{U})=64\,u^{10}(32\,u-27)^2(32\,u^2-67\,u+64)^6(s^6-3\,s^4-8\,us^2+3\,s^2-1),$$
whose zero set (consisting of algebraic curves) divide the parameter
space $\rnum^2$ into $9$ separated regions (see Figure
\ref{fg:par-divide}). By Theorem \ref{thm:main}, for all the points
in each region, $\nz(\sas{U})$ is invariant. We choose $9$ sample
points
\begin{equation*}
  \begin{split}
    &A_1=(-1, -1),~ A_2=(0, -1),~ A_3=(1, -1), ~A_4=(-2, 1/2), \\
    &A_5=(0, 1/2),~ A_6=(2, 1/2),~ A_7=(-3, 1),~ A_8=(0, 1), ~A_9=(3, 1)
  \end{split}
\end{equation*}
as shown in Figure \ref{fg:par-divide}. Let $A_i$ also denote the
corresponding region of the parameter space.
\begin{figure}[h]\centering
  \includegraphics[width=7cm]{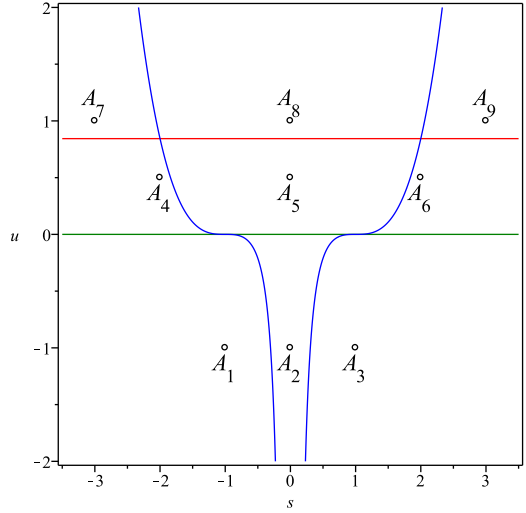}\\
  \caption{Partitions of the parameter space and sample points}\label{fg:par-divide}
\end{figure}

\ste{B4.} The number of real solutions of $\sas{U}$ for each region
$A_i$ can be determined by counting $\nz(\sas{U}|_{A_i})$. For
example, the result of $\sas{U}$ specialized at the sample point
$A_1$ is
\begin{equation*}
\sas{U}|_{A_1}=\left\{
\begin{array}{l}\smallskip
x^6+x^4+18\,x^3+33\,x^2+18\,x+2=0,\\
(x^3+9\,x^2+17\,x+10)(-3\,x^2-8\,x-7) >0.
\end{array}
\right.
\end{equation*}

Using the approach presented in Section \ref{sec:par-alg}, one can
verify that the above parameter-free system has no real solution. So
$\sas{U}$ has no real solution in region $A_1$ (without its border).
In other regions $A_2,\ldots,A_9$, the numbers of distinct real
solutions of $\sas{U}$ can be similarly computed; they are  $1, 2,
0, 1, 2, 0, 1, 2$ respectively.

Thus, provided that $S_1S_2(u^2-2\,u-1)\neq 0$ and $\bp(\sas{U})\neq
0$, or simply
$$N=u(32\,u-27)(u^2-2\,u-1)R\neq 0,$$
where $R=s^6-3\,s^4-8\,us^2+3\,s^2-1$, the number of distinct real
solutions of system \eqref{ex:with-par} is
\begin{itemize}
  \item $0$ if and only if $R<0$ and $s<0$~(i.e., $(s,u)\in A_1\cup A_4\cup A_7$);
  \item $1$ if and only if $R>0$~(i.e., $(s,u)\in A_2\cup A_5\cup A_8$);
  \item $2$ if and only if $R<0$ and $s>0$~(i.e., $(s,u)\in A_3\cup A_6\cup A_9$).
\end{itemize}

In the above result, the additional polynomial $s$ is a must because
the left ($A_1,A_4,A_7$) and right ($A_3,A_6,A_9$) regions cannot be
distinguished by using the sign of $R$ only. For this simple
example, $s$ can be easily observed from Figure \ref{fg:par-divide}.
However, in general it is challenging to find polynomials to
distinguish different regions described by the same inequality. It
is pointed out by Yang and others \cite{y01c} that such polynomials
are contained in the so-called generalized discriminant list and can
be picked out by repeated trials.

For the case when $N=0$, one may add the equation to
\eqref{ex:with-par} and apply the above approach similarly. The
difference is that only one of $u, s$ is now viewed as parameter.
Repeating the process, one can finally obtain the real solution
classification of system \eqref{ex:with-par} for all the points in
the parameter space.

\subsection{General Method}

The problem of analyzing equilibria of an economy with parameters
can be reduced to that of real solution classification of the
following semi-algebraic system
\begin{equation}\label{eq:par-sys}
\left\{
\begin{array}{l}\smallskip
F_1(\uvar,\xvar)=0,\ldots,F_n(\uvar,\xvar)=0,\\ \smallskip
N_1(\uvar,\xvar)\neq 0,\ldots,N_m(\uvar,\xvar)\neq 0,\\ \smallskip
P_1(\uvar,\xvar)> 0,\ldots,P_s(\uvar,\xvar)> 0,\\
P_{s+1}(\uvar,\xvar)\geq 0,\ldots,P_{s+t}(\uvar,\xvar)\geq 0.
\end{array}
\right.
\end{equation}
Let $\pset{F}=[F_1,\ldots,F_n]$, $\pset{N}=[N_1,\ldots,N_m]$, and
the parameters and variables be ordered as
$u_1<\cdots<u_d<x_1<\cdots<x_n$. Our general method for solving the
problem of real solution classification of \eqref{eq:par-sys}
consists of the following main steps.

  \step{B1.}
   Decompose the polynomial system $[\pset{F},\pset{N}]$ into finitely many regular systems $[\tss_1,\sss_1],\ldots,[\tss_k,\sss_k]$ such that
  $$\zero(\pset{F}/\pset{N})=\bigcup_{i=1}^k\zero(\tss_i/\sss_i),$$
  where the zero sets of different main branches do not intersect with each other.
  Without loss of generality, suppose that the first $r$ regular systems $[\tss_1,\sss_1],\ldots,[\tss_r,\sss_r]$ are the main branches.
  We may also suppose that $\tss_1,\ldots,\tss_r$ are all quasi-linear, for otherwise they may be made quasi-linear
  by quasi-linearization using appropriate linear transformations.

  \step{B2.}
  For each $i=1,\ldots,r$, let
  $$\tss_i=[T_{i1}(\uvar,x_1),\dots,T_{in}(\uvar,x_1,\ldots,x_n)]$$
  and $S_i$ be the product of all the polynomials (in $\uvar$) in $\sss_i$. Solve $T_{ij}=0$ for $x_j,~j=n,\ldots,2$,
  and substitute the solutions successively into $P_{l}$,
  $l=1,\ldots,s+t$, to obtain
  rational functions $A_{il}/A'_{il}$ respectively,
  where $A_{il},A'_{il}$ are all polynomials in $x_1$ with parameters $\uvar$.
  Then under the assumption that the parameters $\uvar$ satisfy $S_i\neq 0$ and $\res(A_{il},T_{i1})\neq 0$,
  the problem is further reduced to that of real solution classification of the semi-algebraic system
  \begin{equation*}
    \sas{U}_i=\left\{
    \begin{array}{l}\smallskip
    T_{i1}=0,\\
    A_{il}\cdot A'_{il}>0,\quad l=1,\ldots,s+t.
    \end{array}
    \right.
  \end{equation*}

  \step{B3.} For each $i=1,\ldots,r$, construct the border polynomial $\bp(\sas{U}_i)$, whose
   real zero set decomposes the parameter space into separated regions.
   By Theorem \ref{thm:main}, $\nz(\sas{U}_i)$ is invariant in any region.
  Choose a sample point from each region (which can be done automatically by using, e.g.,
   the method of partial cylindrical algebraic decomposition (PCAD) \cite{c91p}).

 \step{B4.} For each region, determine $\nz(\sas{U}_i)$ by counting the number of distinct real solutions of
 $\sas{U}_i$ at the sample point. Finally, combining the computed results,
 we obtain the necessary and sufficient conditions on $\uvar$ for system \eqref{eq:par-sys}
 to have any given number of distinct real solution, provided that
 $$S_i\cdot\bp(\sas{U}_i)\prod_{l=1}^{s+t}\res(A_{il},T_{i1})\neq 0,\quad i=1,\ldots,r.$$

 \step{B5.} Determine the numbers of distinct real solutions of the
 semi-algebraic systems corresponding to the regular systems $[\tss_{r+1},\sss_{r+1}],\ldots,[\tss_k,\sss_k]$
similarly by regarding some of the parameters as variables. Treat
each of the cases in which $S_i=0$, or $\bp(\sas{U}_i)=0$, or
$\res(A_{il},T_{i1})=0$ for $l=1,\ldots,s+t$ and $i=1,\ldots,r$ by
adding the equality constraint to the original semi-algebraic system
and by taking one of the parameters as variable.

\medskip
The correctness of the above method is guaranteed by Theorem
\ref{thm:main}, and the termination is obvious. Yang and others
\cite{y01c} proposed a more direct method for real solution
classification of semi-algebraic systems {with parameters}.
Their method avoids the process of quasi-linearization of triangular
sets, which is costly when the degrees of the involved polynomials
are high.

\section{Experimental Results}\label{sec:exp}

\subsection{Arms Race Game with Cheap Talk}

The arms race game is originally proposed by  Baliga and
Sj\"ostr\"om \cite{b04a}. In this game, two players simultaneously
and independently choose between building  new weapons
 ($B$) and not building new weapons ($N$). The payoffs of the $i$th
 player are described as follows:
\begin{center}
\begin{tabular}{c|c c}
   & $B$ & $N$ \\
  \hline
  $B$ & $-c_i$ & $m-c_i$ \\
  $N$ & $-d$ & 0 \\
\end{tabular}
\end{center}
In this table, $c_i>0$ is the cost of acquiring new weapons, $m>0$
represents the gain of a player who chooses $B$ while his or her
enemy chooses $N$, and $d>0$ is  the loss of a player when he or she
chooses $N$ and his or her opponent chooses $B$.

The cost $c_i$ is the private information of player $i$, called the
\emph{type} of player $i$. Let each $c_i$ be independent and
identically distributed (i.i.d.) with a continuous cumulative
distribution function $F$, which has compact support $[0,\val{c}]$
with $\val{c}<d$. In addition, $F$ satisfies $F(0) = 0$, $F(\val{c})
= 1$, and $F'(c) >0$ for $0<c<\val{c}$.

To avoid the outcome of arms race, Baliga and Sj\"ostr\"om
introduced the mechanism of cheap talk to the game, which consists
of three stages. In stage zero, nature chooses the types $c_1$ and
$c_2$. In stage one, each player simultaneously announces a message,
conciliatory or aggressive. What the players plan to do in the
future may be read from their messages. In the final stage, the two
players make their decisions ($B$ or $N$) simultaneously according
to the received messages.

The following lemma provides a basis for the analysis of equilibria
(see the original paper \cite{b04a} by Baliga and Sj\"ostr\"om for
its proof).

\begin{lemma}\label{lem:armrace}
 Suppose that $F(c)\,d\geq c$ for all $c\in[0,\val{c}]$.
 For any sufficiently small $m > 0$, there exists a triple $(c_L,c_*,c_H)$ such that
 \begin{equation}\label{eq:armrace}
   \left\{
   \begin{split}
     & [F(c_H)-F(c_L)]c_L=[1-F(c_H)]m,\\
     & [1-2\,F(c_H)+2\,F(c_L)]c_H=F(c_L)\,d,\\
     & [1-F(c_H)](m-c_*)-F(c_L)c_*=-F(c_L)\,d,\\
     & 1>c_H>c_*>c_L>m>0,\\
     & d>0.
   \end{split}
   \right.
 \end{equation}
Moreover, if $m \to 0$, then $c_H\to 0$.
\end{lemma}

Informally speaking, the value of $c_*$ represents the cut-off where
a player is indifferent between $B$ and $N$ in stage one if both
players send a conciliatory message. The value of $c_L$ and $c_H$
indicate the critical point of the type space where a player changes
its message in stage one. Therefore, if both players have a type
exceeding $c_H$, then both of them send a conciliatory message in
stage one and play $N$ in stage two.

One can observe that if $c_H$ tends to $0$, then the probability of
a player with type greater than $c_H$ tends to $1$. By Lemma
\ref{lem:armrace}, the arms race may be avoided with high
probability if the parameter $m$ is sufficiently small. This is the
main result of \cite{b04a}. We restate it as the following theorem.

\begin{theorem}\label{thm:armrace}
Suppose that $F(c)\,d\geq c$ for all $c\in[0,\val{c}]$. Then for any
$\delta> 0$, there is an $\val{m} > 0$ such that for any $m~(0 < m <
\val{m})$, the arms race game with cheap talk has a perfect Bayesian
equilibrium, where $N$ is played with  probability at least
$1-\delta$.
\end{theorem}

Baliga and Sj\"ostr\"om also pointed out that if $m$ is small
enough, then there may exist another equilibrium with cut-off
$(c_L,c_*,c_H)$ satisfying \eqref{eq:armrace}. For this equilibrium,
$c_L\to 0$ and $c_H\to c_M$ as $m\to 0$, where $c_M$ is determined
by $F(c_M)=1/2$. No statement is made in \cite{b04a} about whether
there are other equilibria.

In order to use the computational techniques presented in this paper
for the analysis of equilibria, $F(c)$ must be replaced by a
concrete polynomial function. For the convenience of comparison, we
use the same setting as Kubler and Schmedders \cite{k10t}, i.e.,
$F(c)=c$ and $\val{c}=1$. Then the conditions which the cut-off
$(c_L,c_*,c_H)$ satisfies are reduced to a semi-algebraic system
\begin{equation}\label{eq:arm-race}
\left\{\begin{split}
&P_1=(c_H-c_L)c_L-(1-c_H)m=0,\\
&P_2=(1-2\,c_H+2\,c_L)c_H-c_Ld=0,\\
&P_3=(1-c_H)(m-c_*)-c_Lc_*+c_Ld=0,\\
&1>c_H>c_*>c_L>m>0 ,\\
&d>0,
\end{split}\right.
\end{equation}
where $m,d$ are parameters.

\subsubsection{Analyzing Equilibria Using Triangular Decomposition}

Decomposing the set of the equation polynomials $P_1,P_2,P_3$ in
\eqref{eq:arm-race} into regular systems under the variable ordering
$d<m<c_L<c_*<c_H$, one may obtain $8$ branches
\begin{equation*}
\begin{split}
[\tss_1,\sss_1]=[[&T_1, T_2, T_3], \{m, m+1, d-2\, m-1\}],\\
[\tss_2,\sss_2]=[[&m, c_L, c_*, 2\,c_H^2 -c_H], \{d+1\}],\\
[\tss_3,\sss_3]=[[&d-1, m, c_*-c_L, c_H-c_L], \{c_L\}], \\
[\tss_4,\sss_4]=[[&m+1, c_L-1, d c_*^2+2\, c_*^2-3\,d c_* -c_* +2\, d^2-1, \\
    &c_*c_H+c_H-2\, c_* +d-1], \{d+2, d+1\}], \\
[\tss_5,\sss_5]=[[&d-2\, m-1, 2\,d c_L^2+c_L^2+d m c_L-2\, m c_L -c_L+m, \\
    & -m c_*-c_*-m c_L +dc_L +d m +m, -c_L c_H-m c_H+c_L^2+m],\\
    &\{d-1, d+1, 2\, d+1\}],\\
[\tss_6,\sss_6]=[[&d+2, m+1, c_L-1, 5\, c_*+7, 2 \,c_H+1], \emptyset], \\
[\tss_7,\sss_7]=[[&d+1, m^2+ m, c_L+m, c_*-m, 2\,c_H^2+2 \,m c_H-c_H+m], \emptyset], \\
[\tss_8,\sss_8]=[[&2\,d+1, 4 \,m+3, 7\,c_L-6, 14\,c_*+9, 7\,c_H+1],
\emptyset]
\end{split}
\end{equation*}
with
\[\begin{array}{l}\smallskip
T_1=(d-2\,m-1) c_L^3 + (2\, md+m) c_L^2 + (dm^2 -2\, m^2 - m)c_L +m^2,\\
\smallskip
T_2=(-m -1)c_*-m c_L +d c_L +dm +m, \\
T_3=(-c_L-m) c_H+c_L^2+m.
\end{array}\]
%
%
The main branch $[\tss_1,\sss_1]$ is of our concern: if $d,m$ are
specialized such that $m\neq 0$, $m+1\neq 0$ and $d-2\, m-1\neq 0$,
then the zero set of $\{P_1,P_2,P_3\}$ is the same as that of
$\tss_1=[T_1,T_2,T_3]$. One can see that $\tss_1$ is quasi-linear,
so the number of real zeros of $\tss_1$ equals to that of $T_1$.

Let $P=a_mx^m+\cdots+a_1x+a_0$ be a polynomial in $x$ with
$a_i\in\rnum$ and $a_m\neq 0$. By Descartes' rule \cite{s02s}, the
number of sign changes of the coefficient sequence $a_m,\ldots,a_0$
gives an upper bound for the number of real positive zeros of $P$.
The coefficient sequence of $T_1$ is
$$d-2\,m-1,~ 2\, md+m,~ m^2 d-2\, m^2 - m,~ m^2.$$
The zero set of the polynomials {in this sequence}
decomposes the parameter space into regions as shown in Figure
\ref{fg:armrace}. Since both $m$ and $d$ are required to be positive
real numbers, only the first quadrant of the parameter space need be
considered.
\begin{figure}[h]\centering
  \includegraphics[width=7cm]{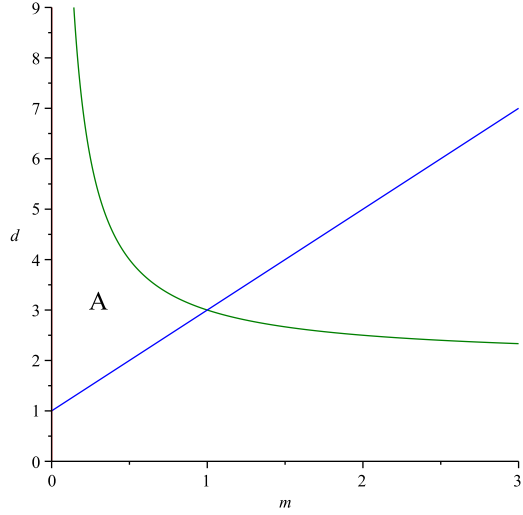}\\
  \caption{Region $A$}\label{fg:armrace}
\end{figure}

In any given region, the number of sign changes of the coefficient
sequence is constant. Baliga and Sj\"ostr\"om were interested mainly
in the case when $m$ is sufficiently small (i.e., the attraction of
building new weapons is not big). For our setting, the pre-condition
$F(c)\,d\geq c$ in Lemma \ref{lem:armrace} and Theorem
\ref{thm:armrace} is reduced to $d\geq 1$. Here we only consider
region $A$ (see Figure \ref{fg:armrace}) of the parameter space,
which can be described by
\begin{equation*}
d-2\,m-1>0,~ 2\, md+m>0,~ m^2 d-2\, m^2 - m<0,~ m^2>0.\\
\end{equation*}
The number of sign changes of $T_1$'s coefficient sequence in region
$A$ is $2$, which is an upper bound for the number of positive real
solutions of system \eqref{eq:arm-race}. Therefore, there exist at
most two equilibria under our setting if $m$ is sufficiently small
and $d\geq 1$ (or $F(c)\,d\geq c$). This demonstrates that the
result of Baliga and Sj\"ostr\"om on the number of equilibria in the
arms-race game is complete.



\subsubsection{Analyzing Equilibria Using Real Solution Classification}


By the method described in Section \ref{sec:par}, the original
system \eqref{eq:arm-race} can be reduced to a semi-algebraic system
in one variable. The squarefree part of the border polynomial of the
reduced system is
\begin{equation*}
B=dm(d-1)(m+1)(2\,d-m-1)(d-2\, m-1)R_1,
\end{equation*}
where
\begin{equation*}
R_1= 8 \,d^3 m^2  -48\, d^2 m^2+ 96\, d m^2-64 m^2  -71\, d^2 m+104\, d m-32\, m+4\, d-4.
\end{equation*}
Note that $R_1>0$ corresponds to two different regions as shown in
Figure~\ref{fg:r1-region} (with red color). Another polynomial $R_2$
is needed for distinguishing the two regions, where
\begin{equation*}
\begin{split}
R_2=&16\,d^2m^4-64\,dm^4+64\,m^4+32\,d^3m^3-20\,d^2m^3-78\,dm^3+64\,m^3\\
&+16\,d^4m^2-36\,d^3m^2+144\,d^2m^2-240\,dm^2+116\,m^2+3\,d^4m\\
&-100\,d^3m+247\,d^2m-206\,dm+56\,m-8\,d^3+24\,d^2-24\,d+8.
\end{split}
\end{equation*}
It is easy to see from Figure \ref{fg:r1-region} that in the first
quadrant $R_2<0$ and $R_2>0$ contain the left and right part of
$R_1>0$ respectively.

\begin{figure}[h]\centering
  \includegraphics[width=7cm]{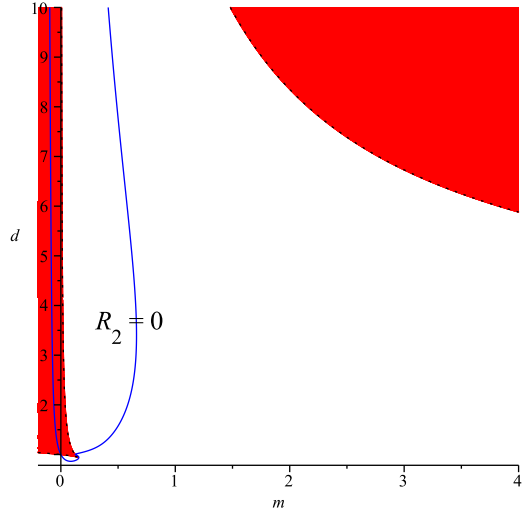}\\
  \caption{Regions of $R_1>0$}\label{fg:r1-region}
\end{figure}

In summary, provided that $B\neq 0$ we can classify the number of
equilibria as follows: under our setting, the arms race game has
\begin{itemize}
  \item    $1$ equilibrium if and only if $d-1<0, 2\,d-m-1>0, R_1<0$;
  \item   $2$ equilibria if and only if $d-1>0, R_1>0, R_2< 0$;
  \item   $3$ equilibria if and only if  $d-1<0, R_1>0$.
\end{itemize}

Assuming that $d\geq 1$ and $m$ is sufficiently small, Baliga and
Sj\"ostr\"om showed the existence of two possible equilibria. The
above result confirms the conclusion of Baliga and Sj\"ostr\"om.
Moreover, we can assert that there is no other equilibrium if $m$ is
small enough and $d>1$ under our setting.

More subtly, we can further compute the region in which there is
only one equilibrium such that $c_H$ is less than a given small
number $a$. We only need to add the inequality $c_H<a$ into system
\eqref{eq:arm-race} and similarly compute the region in which real
solutions exist. Figure \ref{fg:ugoto0} shows the cases in which $a$
equals to $1/10$, $1/20$, and $1/30$. One can see that the
corresponding region shrinks to the $d$-axis as $a\rightarrow 0$.
This confirms the main result of Baliga and Sj\"ostr\"om, i.e.,
Theorem \ref{thm:armrace} in \cite{b04a}. Moreover, the
pre-condition $F(c)\,d\geq c$ ($d\geq 1$) is crucial. For our
setting, if $d\geq 1$ is not satisfied, it can be observed from
Figure~\ref{fg:ugoto0} that no equilibrium with a given small $c_H$
exists even if $m$ is close enough to $0$, which means that the arms
race cannot be avoided with large probability.

\begin{figure}[h]\centering
  \includegraphics[width=7cm]{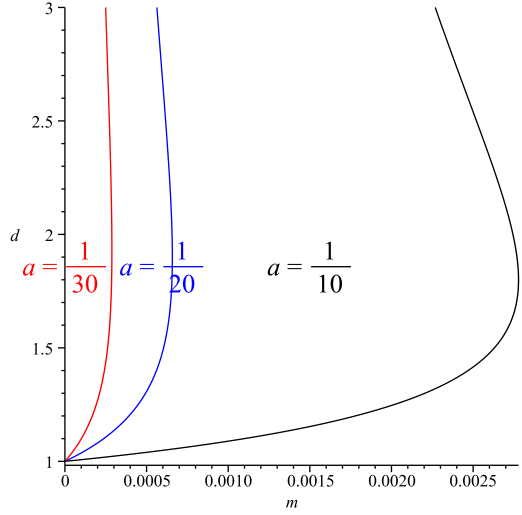}\\
  \caption{Region in which one equilibrium exists with $c_H<a$}\label{fg:ugoto0}
\end{figure}

\subsection{Exchange Economy with Quadratic Utility}

Consider the Arrow--Debreu exchange model with $2$ agents and $2$
commodities, studied first in \cite{k10c}. Let $u_{hl}$ denote the
utility functions for agent $h$ and commodity $l$, where
\begin{equation*}
  \begin{split}
    &u_{11}(c)=9\,c-1/2\,c^2,\quad u_{12}(c)=29/4\,c-7/16\,c^2,\\
    &u_{21}(c)=116\,c-13\,c^2,\quad u_{22}(c)=24\,c-2\,c^2.
  \end{split}
\end{equation*}
Suppose that the endowments of the two agents are $(e_{1},0)$ and
$(0,e_{2})$ respectively, where $e_1,e_2$ are parameters of the
economy. In order to obtain computational results in reasonable
time, we restrict the values of the endowments by $0<e_{h}\leq 10$.

Let $p_1$ and $p_2$ be the prices of the two goods. Moreover,
$(c_{11},c_{12})$ and $(c_{21},c_{22})$ represent the allocations of
the commodities. Then $(p_1,p_2,c_{11}, c_{12},c_{21},c_{22})$ is a
competitive equilibrium of the economy if the allocations maximize
the above utilities under the conditions
\begin{equation*}
  \begin{split}
    &p_1c_{11}+p_2c_{12}\leq p_1e_1,\\
    &p_1c_{21}+p_2c_{22}\leq p_2e_2,\\
    &c_{11}+c_{21}=e_1,\\
    &c_{12}+c_{22}=e_2.
  \end{split}
\end{equation*}

An interior Walrasian equilibrium is a solution $(p_1,p_2,c_{11},
c_{12},c_{21},c_{22},\lambda_1,\lambda_2)$ of the semi-algebraic
system
\begin{equation*}\left\{
  \begin{split}
    &u'_{11}(c_{11})-\lambda_1 p_1=0,\\
    &u'_{12}(c_{12})-\lambda_1 p_2=0,\\
     &u'_{21}(c_{21})-\lambda_2 p_1=0,\\
    &u'_{22}(c_{22})-\lambda_2 p_2=0,\\
    &p_1c_{11}+p_2c_{12}- p_1e_1=0,\\
    &p_1c_{21}+p_2c_{22}-p_2e_2=0,\\
    &c_{11}+c_{21}-e_1=0,\\
    &p_1+p_2-1=0,\\
    &p_1>0,p_2>0,\lambda_1>0,\lambda_2>0,\\
    &c_{hl}>0,10\geq e_{h}>0,
  \end{split}\right.
\end{equation*}
where $u'_{hl}$ (the derivatives of $u_{hl}$) are the marginal
utility functions.

Fix a variable ordering, e.g., $e_1<e_2<p_1<p_2<c_{11}<
c_{12}<c_{21}< c_{22}<\lambda_1<\lambda_2$, and decompose the set of
equation polynomials into regular systems. The result is somewhat
{complicated and} annoying for reading and thus is not
reproduced here. Among the obtained regular systems, the main branch
is quasi-linear and the first polynomial in the regular set is of
degree $4$ in $p_1$. Next compute the equivalent semi-algebraic
system in one variable and construct its border polynomial. The
squarefree part of the border polynomial is of degree $25$ with
$249$ terms.

The final analytical result that can be derived is: multiple
(exactly three) equilibria appear in the trade economy if and only
if $R<0$, where
\begin{equation*}
  \begin{split}
    R=& 14336\,e_2^4-2489600\,e_2^3+3153968\,e_1^2e_2^2-75973600\,e_1e_2^2+603410000\,e_2^2\\
    &-73508800\,e_1^2e_2+1369715000\,e_1e_2-8810812500\,e_2+106496\,e_1^4\\
    &-12416000\,e_1^3+925640000\,e_1^2-13045500000\,e_1+60315234375,\\
  \end{split}
\end{equation*}
provided that the border polynomial is not annihilated.

For any given values $\overline{e}_{1},\overline{e}_{2}\in(0,10]$,
the multiplicity of equilibria can be easily obtained by determining
the sign of $R(\overline{e}_{1},\overline{e}_{2})$. For example,
Kubler and Schmedders \cite{k10c} showed that the economy has $3$
equilibria when $e_1=10$, $e_2=10$. This result can be confirmed by
using our approach since $R(10,10)=-11390625$.

Figure \ref{fg:quadratic} shows the region of the parameter space
described by $R<0$, $0<e_1,e_2\leq 10$. It may be observed that the
possibility for the existence of multiple equilibria is very small.
Furthermore, one can see that multiple equilibria may not appear
when either of the endowment parameters $e_1,e_2$ is sufficiently
small.

\begin{figure}[h]\centering
  \includegraphics[width=7cm]{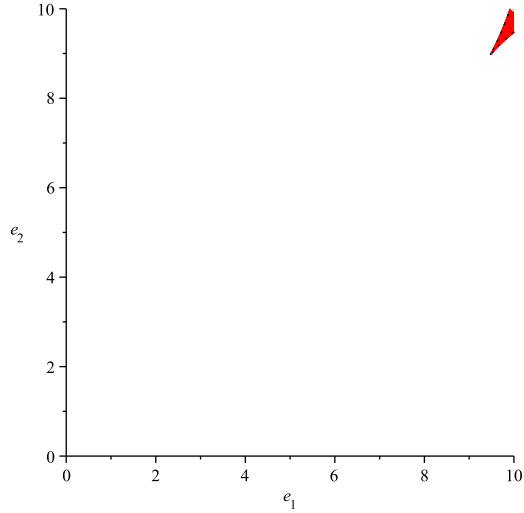}\\
  \caption{Region of $R<0$, $0<e_1,e_2\leq 10$ }\label{fg:quadratic}
\end{figure}

\section{Conclusion}\label{sec:con}

Determining the existence of multiple equilibria is an important
issue in both theoretical and practical studies of economic models.
Equilibria of semi-algebraic economies may be characterized by
semi-algebraic systems. We have proposed an approach that allows
systematic identification of the multiplicities of equilibria of
semi-algebraic economies.

Two problems of identifying multiplicities of equilibria, for
semi-algebraic economies without or with parameters, are addressed.
The basic idea of solving the problems is to first transform the
underlying semi-algebraic systems in several variables into those in
a single variable and then analyze the real solutions of the
resulting systems. The methods we have presented are different from
those based on numerical computation. They can be used to establish
exact and rigorous results and thus are more adequate for the
theoretical study of economic models.

Compared to the method of Kubler and Schmedders \cite{k10c,k10t},
ours can better handle models with inequality constraints, which are
fairly prevalent in practice. Moreover, for parametric economies,
necessary and sufficient conditions for the existence of multiple
equilibria can be automatically generated and explicitly given by
using our methods.

A shortcoming of our methods is their low efficiency for large
problems. We expect to work out specialized and efficient techniques
(e.g., by combining triangular decomposition with PCAD) to improve
the performance of our methods for certain classes of large economic
models. It is hoped that the approach introduced in this paper can
be refined, extended, and further developed to become a potentially
powerful alternative or complement to the widely used numerical
approaches for computational economics.

\bibliographystyle{elsarticle-num}
\bibliography{CompEco}

\end{document}